\documentclass[a4paper,envcountsame]{llncs}

\usepackage{amsmath}
\usepackage{amssymb}
\usepackage{tabularx}

\newcommand{\F}{\mathcal{F}}
\newcommand{\np}{\mathsf{NP}}
\newcommand{\pair}[2]{(#1,#2)}

\newcommand{\prob}[3]{
\begin{center}\fbox{\parbox{0.9\linewidth}{
#1\\
\begin{tabularx}{\linewidth}{rX}
\textbf{Input:} & #2\\
\textbf{Question:} & #3
\end{tabularx}
}}
\end{center}
}
\begin{document}

\title{A faster FPT algorithm for Bipartite Contraction\thanks{Research  supported in part by the European Research Council (ERC)  grant 
``PARAMTIGHT: Parameterized complexity and the search for tight
complexity results,'' reference 280152 and OTKA grant NK105645.}}
\author{Sylvain Guillemot and D\'{a}niel Marx}
\institute{Institute for Computer Science and Control, Hungarian Academy of Sciences\\
(MTA SZTAKI), Budapest, Hungary}
\maketitle{}

\begin{abstract} The \textsc{Bipartite Contraction} problem is to
  decide, given a graph $G$ and a parameter $k$, whether we can can
  obtain a bipartite graph from $G$ by at most $k$ edge
  contractions. The fixed-parameter tractability of the problem was
  shown by Heggernes et al. \cite{HHLP11}, with an algorithm whose
  running time has double-exponential dependence on $k$.  We present a
  new randomized FPT algorithm for the problem, which is both
  conceptually simpler and achieves an improved $2^{O(k^2)} n m$
  running time, i.e., avoiding the double-exponential dependence on $k$. The
  algorithm can be derandomized using standard techniques.
\end{abstract}

\section{Introduction}\label{sec:introduction}

A graph modification problem aims at transforming an input graph into
a graph satisfying a certain property, by at most $k$
operations. These problems are typically studied from the viewpoint of
fixed-parameter tractability, where the goal is to obtain an algorithm
with running time $f(k) n^c$ (or \emph{FPT algorithm}). Here, $f(k)$
is a computable function depending only on the parameter $k$, which
confines the combinatorial explosion that is seemingly inevitable for
an $\np$-hard problem. The most intensively studied graph modification
problems involve vertex- or edge-deletions as their base operation;
fixed-parameter tractability has been established for the problems of
transforming a graph into a forest \cite{GGHNW06,DFLRS07,CFLLV08}, a
bipartite graph \cite{RSV04,GGHNW06,LSS09,KR10,DBLP:journals/corr/abs-1203-0833,DBLP:journals/corr/abs-1304-7505,DBLP:journals/corr/IwataOY13}, a chordal
graph \cite{M10,MR2000c:68068}, a planar graph \cite{MS12,DBLP:conf/focs/Kawarabayashi09}, a unit/proper interval graph \cite{DBLP:conf/wg/BevernKMN10,DBLP:journals/algorithmica/HofV13}, or an interval graph
\cite{DBLP:journals/siamcomp/VillangerHPT09,DBLP:journals/corr/abs-1211-5933}. Results
have also been obtained for problems involving directed graphs
\cite{CLLSR08} or group-labeled graphs \cite{G11,CPP12}.

Recently, there has been an interest in graph modification problems involving edge contractions. These problems fall in the following general framework. Given a graph property $\Pi$, the problem \textsc{$\Pi$-Contraction} is to decide, for a graph $G$ and a parameter $k$, whether we can obtain a graph in $\Pi$, by starting from $G$ and performing at most $k$ edge contractions. For each graph property $\Pi$ admitting a polynomial recognition algorithm, it is then natural to ask whether \textsc{$\Pi$-Contraction} admits an FPT algorithm. Such algorithms have been given when $\Pi$ is the class of paths, the class of trees \cite{HHLLP11}, the class of planar graphs \cite{GHP12}, or the class of bipartite graphs \cite{HHLP11}.

For the case of bipartite graphs, the problem is called \textsc{Bipartite Contraction}, and Heggernes et al.~\cite{HHLP11} obtained an FPT algorithm with a running time double-exponential in $k$. The algorithm combines several tools from parameterized algorithmics, such as the irrelevant vertex technique, important separators, treewidth reduction, and Courcelle's theorem. In this note, we present a new FPT algorithm for the problem, which is both conceptually simpler and faster. Similar to the compression routine for \textsc{Odd Cycle Transversal} in \cite{RSV04}, we reduce \textsc{Bipartite Contraction} to several instances of an auxiliary cut problem. Our main effort is spent on obtaining an FPT algorithm for this cut problem. This is achieved by using the notion of important separators from \cite{M06}, together with the randomized coloring technique introduced by Alon et al~\cite{AYZ95}. We obtain the following result:

\begin{theorem} \label{thm:algo-bc} \textsc{Bipartite Contraction} has
  a randomized \textup{FPT} algorithm with running time $2^{O(k^2)}
  n m$ and a deterministic algorithm with running time $2^{O(k^2)}
  n^{O(1)}$.
\end{theorem}

This paper is organized as follows. We first introduce the relevant notation and definitions in Section \ref{sec:preliminaries}. We explain in Section \ref{sec:reduction-rank-cut} how \textsc{Bipartite Contraction} can be reduced to several instances of a suitable cut problem called \textsc{Rank-Cut}. In Section \ref{sec:constrained-rank-cut}, we define a constrained version of the \textsc{Rank-Cut} problem and show that it is polynomial-time solvable. In Section \ref{sec:reduction-constrained-rank-cut}, we present a randomized reduction of \textsc{Rank-Cut} to its constrained version. Finally, in Section \ref{sec:derandomization} we derandomize this reduction and we complete the proof of Theorem \ref{thm:algo-bc}. Section \ref{sec:conclusion} concludes the paper.

\section{Preliminaries}\label{sec:preliminaries}

Given a graph $G$, we let $V(G)$ and $E(G)$ denote its vertex set and
edge set, respectively. We let $n = |V(G)|$ and $m = |E(G)|$.  Given
$X \subseteq V(G)$, we denote by $G[X]$ the subgraph of $G$ induced
by $X$, and we denote by $G \setminus X$ the subgraph of $G$ induced
by $V(G) \setminus X$. Given a set $F\subseteq E(G)$ of edges, we denote
by $V(F)$ the endpoints of the edges in $F$, and we say that $F$ \emph{spans}
the vertices in $V(F)$. Given $F \subseteq E(G)$, we denote by $G[F]$ the graph with
vertex set $V(F)$ and edge set $F$; we denote by $G \setminus F$ the graph with
vertex set $V(G)$ and edge set $E(G) \setminus F$. For an edge $e$, we denote by $G/e$
the graph obtained by contracting edge $e$, that is, by removing the
endpoints $u$ and $v$ of $e$ and introducing a new vertex that is adjacent to
every vertex that is adjacent to at least one of $u$ or $v$. Given $F
\subseteq E(G)$, we denote by $G / F$ the graph obtained from $G$ by
contracting all the edges of $F$; it is easy to observe that the graph
$G/F$ does not depend on the order in which we perform the
contractions.

Fix two disjoint subsets of vertices $X,Y$ of a graph $G$. An \emph{$\pair{X}{Y}$-walk} in $G$
is a sequence $W = v_0 v_1 \ldots v_{\ell}$ of vertices such that $v_0 \in X, v_{\ell} \in Y$, and
$v_i v_{i+1} \in E(G)$ for $0 \leq i < \ell$; the \emph{length} of $W$ is $\ell$, and we call $W$
an \emph{$\pair{X}{Y}$-path} in $G$ if the vertices $v_i$ are pairwise distinct. We will simply
use the term ``path'' when the sets $X,Y$ are irrelevant. An \emph{$\pair{X}{Y}$-cut} in $G$ is
defined as a set $C \subseteq E(G)$ such that $G \setminus C$ has no $\pair{X}{Y}$-path; an
\emph{$\pair{X}{Y}$-separator} in $G$ is defined as a set $S \subseteq V(G)
\setminus (X \cup Y)$ such that $G \setminus S$ has no $\pair{X}{Y}$-path. Note that
the $\pair{X}{Y}$-separator is by definition disjoint from $X$ and $Y$. An
$\pair{X}{Y}$-cut (resp.~$\pair{X}{Y}$-separator) $C$ is \emph{inclusion-wise minimal} if no
proper subset of $C$ is an $\pair{X}{Y}$-cut (resp.~$\pair{X}{Y}$-separator).

A \emph{bipartite modulator} of $G$ is a set $F \subseteq E(G)$ such
that $G \setminus F$ is bipartite. The \emph{rank} of a graph $G$ is the number of edges
of a spanning forest of $G$. The rank of a set
$F\subseteq E(G)$ of edges, denoted by $r(F)$ is the rank of $G[F]$. As observed in \cite{MSR12}, we can
alternatively define \textsc{Bipartite Contraction} as the following
problem: given a graph $G$ and an integer $k$, find a bipartite
modulator $F$ of $G$ such that $r(F) \leq k$. We reproduce the proof
here for completeness.

\begin{lemma} \label{lem:equivalence-bc} The following statements are equivalent:
\begin{enumerate}
\item[(i)] there exists a set $F \subseteq E(G)$ such that $|F| \leq k$ and $G / F$ is bipartite; 
\item[(ii)] there exists a set $F \subseteq E(G)$ such that $r(F) \leq k$ and $G \setminus F$ is bipartite.
\end{enumerate}
\end{lemma}
\begin{proof} $(i) \Rightarrow (ii)$: Let $F'$ denote the edges of $G$ having both endpoints in a same connected component of $G[F]$. Observe that $r(F') \leq k$, as $F$ contains a spanning forest of $G[F']$. We claim that $G \setminus F'$ is bipartite. Observe that the vertex set of each connected component of $G[F]$ is an independent set in $G \setminus F'$. Therefore, a proper 2-coloring of $G / F$ can be turned into a proper 2-coloring of $G \setminus F'$ if we color every vertex in a connected component $K$ of $G[F]$ by the color of the single vertex corresponding to $K$ in $G / F$.

$(ii) \Rightarrow (i)$: Let us fix a proper 2-coloring of $G \setminus F$. We can assume that $F$ is a minimal set of edges such that $G \setminus F$ is bipartite. Therefore, in each connected component of $G[F]$, every vertex has the same color in the 2-coloring of $G \setminus F$. Hence, contracting each connected component of $F$ to a single vertex gives a bipartite graph. This graph can be obtained by contracting the edges of a spanning forest of $F$, which has $r(F) \leq k$ edges. \qed
\end{proof}

One can deduce from Lemma \ref{lem:equivalence-bc} that \textsc{Bipartite Contraction} is a monotone problem in the sense that positive instances are preserved by taking (not necessarily induced) subgraphs. Such a monotonicity property is typically required when the problem is solved using the iterative compression technique (see, e.g., \cite{GGHNW06,HHLP11}). As we shall see in Section \ref{sec:reduction-rank-cut}, even though we solve the problem by ``compressing'' from a set of vertices $X$ whose deletion makes the graph bipartite, we do not use \emph{iterative} compression as we perform a single compression step from a set $X$ obtained with a black-box algorithm. Hence, we do not explicitly need the monotonicity property in this paper.

\section{Reduction to a cut problem}\label{sec:reduction-rank-cut}

We first define a \emph{compression} version of the problem, named \textsc{Bipartite Contraction Compression}: given a graph $G$, an integer $k$, and a set $X \subseteq V(G)$ with $|X|\le 2k$ such that $G \setminus X$ is bipartite, is there a bipartite modulator $F$ of $G$ with $r(F) \leq k$? The following lemma establishes how a \emph{compression routine} for the problem entails the fixed-parameter tractability of \textsc{Bipartite Contraction}.

\begin{lemma} \label{lem:reduction-bcc} Suppose that \textsc{Bipartite Contraction Compression} is solvable in time $T(k,n,m)$. Then \textsc{Bipartite Contraction} is solvable in time $O(9^k k n m + T(k,n,m))$.
\end{lemma}

\begin{proof} An instance $I = (G,k)$ of \textsc{Bipartite
    Compression} is solved the following way. First, we run the
  algorithm of Reed et al.~\cite{RSV04} to look for a set $X \subseteq
  V(G)$ of size $\leq 2k$ such that $G \setminus X$ is bipartite; the running
  time of the algorithm is $O(3^{2k} k nm)$ (see also
  \cite{LSS09})\footnote{Very recently, two linear-time algorithms for \textsc{Odd Cycle Transversal} were announced \cite{DBLP:journals/corr/abs-1304-7505,DBLP:journals/corr/IwataOY13}. Also, the dependence on $k$ was improved from $3^k$ to $2.3146^k$ \cite{DBLP:journals/corr/abs-1203-0833}. However, using any of these algorithms here would not improve the overall asymptotic running time of our algorithm, as the dominating term comes from the \textsc{Rank-Cut} algorithm of Theorem~\ref{thm:rank-cut-deterministic}. }. If there is no such set, then we answer
  ``no''. Otherwise, we run the algorithm for \textsc{Bipartite
    Contraction Compression} on $(G,k,X)$. This takes time $O(9^k k n m +
  T(k,n,m))$ as claimed. The correctness of this algorithm
  follows by observing that, if $F$ is a solution for instance $I$ of
  \textsc{Bipartite Contraction}, then $X = V(F)$ is a set of size at
  most $2k$ such that $G \setminus X$ is bipartite. \qed
\end{proof}

In the rest of this section, we concentrate on the \textsc{Bipartite Contraction Compression} problem. We solve the problem similarly to the compression routine for \textsc{Odd Cycle Transversal} \cite{RSV04}. First, we adapt the construction of \cite{RSV04} to the case of edge sets.

Suppose that we are given a graph $G$ in which a bipartite modulator has to be found,
along with a set $X \subseteq V(G)$ such that $G \setminus X$ is bipartite.
We construct a graph $G'$ as follows. Let $S_1, S_2$ be a bipartition of $G \setminus X$,
and let $<$ be an arbitrary total ordering of $V(G)$. We let $V(G') = (V(G) \setminus X)
\cup X'$ with $X' = \{ x_1,x_2 : x \in X \}$, and $E(G') = E(G \setminus X) \cup \{ u v_{3-i} :
uv \in E, u \in S_i, v \in X \} \cup \{ u_1 v_2 : uv \in E, u,v \in X, u < v \}$. Observe that $G'$
is bipartite, with bipartition $S'_1 = S_1 \cup \{ x_1 : x \in X \}$ and $S'_2 = S_2 \cup
\{ x_2 : x \in X \}$. Furthermore, if we identify $x_1$ with $x_2$ for every $x\in X$
in $G'$, then we get the graph $G$; in particular, $G$ and $G'$ have
the same number of edges.

Define a bijection $\Phi : E(G) \rightarrow E(G')$ which preserves
each edge of $E(G \setminus X)$, maps each edge $uv \in E(G)$ ($u \in S_i, v \in
X$) to the edge $u v_{3-i}$, and maps each edge $uv \in E(G)$ ($u,v \in
X, u < v$) to the edge $u_1 v_2$.  We say that a partition of $X'$
into two sets $X'_A,X'_B$ is \emph{valid} if for each $x \in X$,
exactly one of $\{x_1,x_2\}$ is in $X'_A$.  The following lemma is
similar to Lemma 1 of \cite{RSV04}.

\begin{lemma} \label{lem:equivalence-RSV} For every $F \subseteq E(G)$, the following statements are equivalent:
\begin{enumerate}
\item[(i)] $G \setminus F$ is bipartite,
\item[(ii)] There is a valid partition $X'_A,X'_B$ of $X'$ such that $\Phi(F)$ is an $\pair{X'_A}{X'_B}$-cut in $G'$.
\end{enumerate}
\end{lemma}

\begin{proof} $(i) \Rightarrow (ii)$. Suppose that $G \setminus F$ is
  bipartite with bipartition $V_1,V_2$. Define the partition
  $X'_A,X'_B$ of $X'$ such that: for $u \in X$, $u_1 \in X'_A$ iff $u
  \in V_1$. Observe that $X'_A,X'_B$ is a valid partition of $X'$. We
  claim that $C = \Phi(F)$ is an $\pair{X'_A}{X'_B}$-cut in $G'$.
  Towards a contradiction, assume that $G' \setminus C$ contains an
  $X'_A,X'_B$-path $P'$. Suppose that the endpoints of $P'$ are $u_i \in
  X'_A, v_j \in X'_B$; then $u \in V_i, v \in V_{3-j}$. The path $P'$
  corresponds to an $u,v$-path $P$ in $G \setminus F$, of the same
  length.  If $j = i$, then $u_i,v_j$ both belong to $S'_i$, and we
  obtain that $P$ is a path of even length between $V_i$ and
  $V_{3-i}$, contradiction. If $j = 3-i$, then $u_i \in S'_i$, $v_j
  \in S'_{3-i}$, and we obtain that $P$ is a path of odd length
  between $V_i$ and $V_i$, contradiction. We conclude that $C$ is an
  $\pair{X'_A}{X'_B}$-cut in $G'$, as claimed.

  $(ii) \Rightarrow (i)$. Suppose that $C \subseteq E(G')$ is an
  $\pair{X'_A}{X'_B}$-cut in $G'$, for some valid partition $X'_A,X'_B$ of
  $X'$. We claim that $F = \Phi^{-1}(C)$ is such that $G \setminus
  F$ is bipartite. We define a 2-coloring $\chi$ of $G \setminus F$
  as follows: (1) If $u \in X$, then $\chi(u) = 1$ iff $u_1 \in X'_A$;
  (2) If $u \in V \setminus X$ is reachable from $X'_A$ in $G'
  \setminus C$, then $\chi(u) = 1$ iff $u \in S_1$; (3) If $u \in V
  \setminus X$ is not reachable from $X'_A$ in $G' \setminus C$,
  then $\chi(u) = 1$ iff $u \in S_2$. We verify that $\chi$ is a
  proper 2-coloring of $G \setminus F$. Consider an edge $uv \in E(G)
  \setminus F$, there are three cases. If $u,v \notin X$, then $u \in
  S_i, v \in S_{3-i}$; as $u,v$ are either both in case (2) or both in
  case (3), it follows that $\chi(u) \neq \chi(v)$. If $u \in S_i, v
  \in X$, then $u v_{3-i}$ is an edge of $G' \setminus C$; if
  $v_{3-i} \in X'_A$ then $\chi(v) = 3-i$ by (1), and $\chi(u) = i$ by
  (2); if $v_{3-i} \in X'_B$ then $\chi(v) = i$ by (1), and $\chi(u) =
  3-i$ by (3). If $u,v \in X$ with $u < v$, then $u_1 v_2$ is an edge
  of $G' \setminus C$, and thus we have $u_1,v_2$ both in $X'_A$ or
  both in $X'_B$, which implies that $\chi(u) \neq \chi(v)$. We
  conclude that $G \setminus F$ is bipartite, as claimed. \qed
\end{proof}

Lemma~\ref{lem:equivalence-RSV} turns the problem of finding a
bipartite modulator into several instances of a cut problem (one for
each valid partition).  The same way as Lemma~\ref{lem:equivalence-bc}
shows the equivalence of \textsc{Bipartite Contraction} with the problem of finding a bipartite modulator with a rank
constraint, we show that Lemma~\ref{lem:equivalence-RSV} allows us to
solve \textsc{Bipartite Contraction Compression} by solving a cut
problem with a rank constraint. However, there is a technical detail
related to the fact that two vertices $x_1,x_2\in X'$ correspond to
each vertex $x\in X$ in the construction of $G'$; we need the
following definition to deal with this issue.
Let $M,F \subseteq E(G)$ be two subsets of edges. We define the \emph{$M$-rank} $r_{M}(F)$ of $F$ as the rank of the graph $G[F \cup M] / M$. Our auxiliary problem is defined as follows.
\medskip

\prob{\textsc{Rank-Cut}}{A graph $G$, an integer $k$, two sets $X,Y \subseteq V(G)$, and a set $M\subseteq E(G)$ with $|M| \leq 2k$}{Is there an $\pair{X}{Y}$-cut $C$ in $G$ such that $r_{M}(C) \leq k$?}
\medskip

The following simple observation will be useful later:
\begin{lemma}\label{lem:solspan}
If $|M|\le 2k$ and $r_M(C)\le k$, then $C\cup M$ spans at most $6k$ vertices.
\end{lemma}
\begin{proof}
Each contraction can decrease rank by at most one, hence the rank of $G[C\cup M]$ is at most $3k$. As $G[C\cup M]$ has no isolated vertices by definition, it follows that $G[C\cup M]$ has at most $6k$ vertices. \qed
\end{proof}

We now describe how an FPT algorithm for \textsc{Rank-Cut} yields an
FPT algorithm for \textsc{Bipartite Contraction Compression};
Section~\ref{sec:constrained-rank-cut}--\ref{sec:derandomization} show
the fixed-parameter tractability of \textsc{Rank-Cut} itself.
 
\begin{lemma} \label{lem:reduction-rank-cut} Suppose that \textsc{Rank-Cut} is solvable in time $T(k,n,m)$. Then \textsc{Bipartite Contraction Compression} is solvable in $O(4^{k} (T(k,n,m)+n+m))$ time.
\end{lemma}

\begin{proof} Consider an instance $I = (G,k,X)$ of \textsc{Bipartite
    Contraction Compression}. From $G$ and $X$, we construct the graph
  $G'$ as described above Lemma~\ref{lem:equivalence-RSV}. We let $H$
  be obtained from $G'$ by adding the edge $x_1x_2$ for every $x\in X$; let $M\subseteq E(H)$ be 
the set of these new edges.

We solve \textsc{Bipartite Contraction Compression} by the following
algorithm. For each valid partition $X'_A,X'_B$ of $X'$, we run the
algorithm for \textsc{Rank-Cut} on the instance $I' =
(H,k,X'_A,X'_B,M)$. We answer ``yes'' if and only if one of the
instances $I'$ was a yes-instance of \textsc{Rank-Cut}. Note that, as
$|X| \leq 2k$ by assumption, we have $|M| \leq 2k$ and thus each
instance $I'$ is a valid instance of \textsc{Rank-Cut}. As there are
$2^{|X|} \leq 4^k$ valid partitions of $X'$, the claimed running time
follows. We show that it correctly solves \textsc{Bipartite
  Contraction Compression}.

Suppose that $I$ admits a solution $F$ with $r(F) \leq k$, then $G
\setminus F$ is bipartite. Thus, by Lemma \ref{lem:equivalence-RSV}
there exists a valid partition $X'_A,X'_B$ of $X'$ such that $\Phi(F)$ is an
$\pair{X'_A}{X'_B}$-cut in $G'$. Hence, $C = \Phi(F) \cup M$ is an
$\pair{X'_A}{X'_B}$-cut in $H$, and since $H[C]/ M$ is isomorphic to $G[F]$,
we have $r_{M}(C) = r(F) \leq k$. It follows that $C$ is a solution
for \textsc{Rank-Cut} on the instance $I' =
(H,X'_A,X'_B,M,k)$. Conversely, suppose that $C$ is a solution of
\textsc{Rank-Cut} on the instance $I' = (H,X'_A,X'_B,M,k)$, for some
valid partition $X'_A,X'_B$ of $X'$. Observe that $M \subseteq C$ (as
each edge of $M$ is between $X'_A$ and $X'_B$), and that $C
\setminus M$ is an $\pair{X'_A}{X'_B}$-cut in $H\setminus M=G'$. Thus, if we
define $F = \Phi^{-1}(C \setminus M)$, we obtain that $G \setminus
F$ is bipartite by Lemma \ref{lem:equivalence-RSV}. Observe that
contracting the edges of $M$ in $H[C]$ gives a graph isomorphic to
$G[F]$. Therefore, $r(F) = r_{M}(C) \leq k$, and we conclude that $F$ is
a solution for the instance $I$. \qed
\end{proof}

\section{Solving a constrained version of \textsc{Rank-Cut}}\label{sec:constrained-rank-cut}

In this section, we introduce a constrained variant of 
\textsc{Rank-Cut}, and show its polynomial-time
solvability. We give  a randomized reduction
of \textsc{Rank-Cut} to this variant in the next section. In the constrained problem, the
cut has to be constructed as the union of disjoint components prescribed in the input:
\medskip

\prob{\textsc{Constrained Rank-Cut}}{A graph $G$, an integer $k$, two subsets $X,Y \subseteq V(G)$, a set $M\subseteq E(G)$, and a partition
$P = (V_1,\ldots,V_{\ell})$ of $V(G)$ such that
\begin{itemize}
\item[(i)]$G[V_i]$ is connected for every $1\le i \le \ell$, and
\item[(ii)] there is no edge of $M$ between $V_i$ and $V_j$ for any $i\neq j$.
\end{itemize}
}{Is there a set $Z \subseteq \{1,\ldots,\ell\}$ such that $C_Z = \cup_{i \in Z} E(G[V_i])$ is an $\pair{X}{Y}$-cut in $G$ with $r_{M}(C_Z) \leq k$?}
\medskip

Note that a $V_i$ can consist of a single vertex, in which case $E(G[V_i])=\emptyset$ and it does not matter if $i$ is in $Z$ or not.
We show that the constrained problem can be reformulated as a weighted separator problem between two sets, and hence can we solved in polynomial time.

\begin{lemma} \label{lem:algo-constrained-rank-cut} \textsc{Constrained Rank-Cut} can be solved in $O(k(n+ m))$ time.
\end{lemma}

\begin{proof} Let $I = (G,k,X,Y,M,P)$ be an instance of \textsc{Constrained Rank-Cut} with $P = (V_1,\ldots,V_{\ell})$. 
 Starting with $G$, we build a weighted graph $G'$ as follows:
\begin{itemize}
\item we remove the edges of $\cup_{i = 1}^{\ell} E(G[V_i])$;
\item we give an infinite weight to the vertices of $V(G)$;
\item for each $1 \leq i \leq \ell$, we add a vertex $v_i$ of weight $r_{M}(E(G[V_i]))$, and we make $v_i$ adjacent to the vertices of $V_i$.
\end{itemize}
We answer ``yes'' if and only if $G'$ has an $\pair{X}{Y}$-separator of
weight at most $k$. We claim that this algorithm takes $O(k(n+ m))$ time.
First, observe that $G'$ has at most $n+m$ edges: for each $i \in \{1,\ldots,\ell\}$,
we replace the edges of $E(G[V_i])$ by a number of edges equal to $|V_i| \leq
|E(G[V_i])|+1$. As we are trying to find an $\pair{X}{Y}$-separator of
weight at most $k$ in $G'$, we can accomplish this by performing at
most $k$ rounds of the Ford-Fulkerson max-flow min-cut algorithm,
giving the running time $O(k(n+m))$.

Given a set $Z \subseteq \{1,\ldots,\ell\}$, let us define edge set $C_Z = \cup_{i \in Z} E(G[V_i])$ and vertex set $S_Z = \{ v_i : i \in Z \}$. The following claim establishes the correctness of the algorithm.

\begin{quote}
{\em Claim.} For any $Z \subseteq \{1,\ldots,\ell\}$,
\begin{enumerate}
\item[(i)] $r_{M}(C_Z)$ equals the weight of $S_Z$;
\item[(ii)] $C_Z$ is an $\pair{X}{Y}$-cut in $G$ iff $S_Z$ is an $\pair{X}{Y}$-separator in $G'$.
\end{enumerate}
\end{quote}

To prove (i), note first that the vertex set of each connected component of $G[C_Z]$ is some $V_i$. Furthermore, as the two endpoints of each edge in $M$ is in the same $V_i$, it is also true in the graph $G[C_Z \cup M]$ that the vertex set of each connected component is some $V_i$. Thus, each connected component of $G[C_Z \cup M] / M$ is obtained from a set $V_i$ by identifying some vertices. We obtain that $r_{M}(C_Z) = \sum_{i \in Z}  r_{M}(E(G[V_i]))$ equals the weight of $S_Z$.

To prove (ii), suppose that $C_Z$ is an $\pair{X}{Y}$-cut in $G$; we need to
show that $S_Z$ is an $\pair{X}{Y}$-separator in $G'$. By way of
contradiction, assume that $G'$ contains an $\pair{X}{Y}$-path $W$ avoiding
$S_Z$. For each segment of $W$ of the form $x v_i y$ with $x,y \in
V_i, i \notin Z$, we replace it by an $x,y$-path in $G[V_i]$ (recall
that the neighbors of $v_i$ are in $V_i$). We obtain an $\pair{X}{Y}$-walk in
$G$ avoiding $C_Z$, a contradiction.

Conversely, suppose that $S_Z$ is an $\pair{X}{Y}$-separator in $G'$,
and let us show that $C_Z$ is an $\pair{X}{Y}$-cut in $G$. By way of
contradiction, assume that $G$ contains an $\pair{X}{Y}$-path $W$
avoiding $C_Z$. Then $W$ can be partitioned as $W_1 W_2 \ldots W_r$,
where each $W_j$ is a maximal subpath of $W$ included in a set $V_i$
(possibly, $W_j$ contains a single vertex). Each $W_j$ that has at
least two vertices is an $x,y$-path included in a component $V_{i}$
with $i \notin Z$; we replace it by a path of the form $x v_{i} y$, to
obtain an $\pair{X}{Y}$-walk in $G'$ avoiding $S_Z$, a
contradiction. \qed
\end{proof}

\section{Reduction to the constrained version}\label{sec:reduction-constrained-rank-cut}

In this section, we describe a randomized reduction mapping an instance $I = (G,k,X,Y,M)$ of \textsc{Rank-Cut} to an instance $I' = (G,k,X,Y,M,P)$ of \textsc{Constrained Rank-Cut}.

The first step of the reduction identifies a set of \emph{relevant
  edges} $E_{\textup{rel}} \subseteq E(G)$ that spans a graph of bounded
degree. It relies on the notion of important separators introduced in
\cite{M06}, which we recall now. Fix two disjoint sets $X,Y \subseteq
V(G)$, and let $S$ be an $\pair{X}{Y}$-separator in $G$. We denote by
$\textup{Reach}_G(X,S)$ the set of vertices of $G$ reachable from $X$ in $G
\setminus S$; note that $\textup{Reach}_G(X,S)$ is disjoint from $Y$. We say
that $S$ is an \emph{important} $\pair{X}{Y}$-separator if (i) $S$ is an
inclusion-wise minimal $\pair{X}{Y}$-separator, (ii) there is no
$\pair{X}{Y}$-separator $S'$ with $|S'| \leq |S|$ and $\textup{Reach}_G(X,S) \subset
\textup{Reach}_G(X,S')$. We have the following result:

\begin{lemma}[\cite{M06,DBLP:journals/algorithmica/ChenLL09,dirmway-journal}]\label{lem:enum-impseps} Let $k$ be a nonnegative
  integer. There are at most $4^k$ important $\pair{X}{Y}$-separators
  of size $\leq k$, and they can be enumerated in time $O(4^k k
  (n+m))$.
\end{lemma}

We now describe the construction of the set
$E_{\textup{rel}}$. Starting with $G$, we construct a graph $G'$ by
subdividing each edge $e$ with a vertex $z_e$. Given two subsets $X,Y
\subseteq V(G)$, we denote by $C_k(X,Y)$ the union of the important
$\pair{X}{Y}$-separators of size at most $k$ in the extended graph
$G'$. As there are are at most $4^k$ such separators by Lemma~\ref{lem:enum-impseps}, we have $|C_k(X,Y)|\le k\cdot 4^k$. Given a vertex $u \in V(G)$,
we denote by $E(u)$ the set of edges of $G$ incident to $u$. We define
the set $E_{\textup{rel}} \subseteq E(G)$ as follows: (i) for each $u \in V(G)$,
let $E_{\textup{rel}}(u) = \{ e \in E(u) : z_e \in C_{6k}(X,\{u\}) \cup
C_{6k}(Y,\{u\})\}$; (ii) $E_{\textup{rel}}$ consists of those edges $uv
\in E(G)$ such that $uv \in E_{\textup{rel}}(u) \cap
E_{\textup{rel}}(v)$. By Lemma \ref{lem:enum-impseps},
$E_{\textup{rel}}$ can be constructed in time $O(4^{6 k} k \cdot n(n+
m))$, as we need to enumerate important separators for $n$ vertices. Furthermore,
the graph $G[E_{\textup{rel}}]$ has maximum degree $d = 12 k \cdot 4^{6k}$, as each set $E_{\textup{rel}}(u)$ comes from the union of two sets $C_{6k}(X,\{u\})$ and $C_{6k}(Y,\{u\})$, each of which has size at most $6k \cdot 4^{6k}$. The interest of the set $E_{\textup{rel}}$ is that it
contains any minimal solution for $I$.

\begin{lemma} \label{lem:relevant-edges} Any minimal solution $C$ of a
  \textsc{Rank-Cut} instance $I$ is included in $E_{\textup{rel}}$.
\end{lemma}

\begin{proof} We show that for every $e = uv \in C$, we have $e \in
  E_{\textup{rel}}(v)$; this will imply that $e \in
  E_{\textup{rel}}(u)$ by symmetry, and thus $e \in
  E_{\textup{rel}}$. As $C$ is a minimal $\pair{X}{Y}$-cut,
if we define $U$ to be the set of vertices
  reachable from $X$ in $G \setminus C$, then 
  $X\subseteq U \subseteq V(G)\setminus Y$ holds and $C$ is the set of
  edges with exactly one endpoint in $U$. Let $C_X$ denote the
  endpoints of $C$ inside $U$, and let $C_Y$ denote the endpoints of
  $C$ inside $V(G) \setminus U$. We suppose that $v \in C_Y$, as the
  case $v \in C_X$ is similar. Let us define the vertex set $S$ of $G'$ as
  $S = (C_Y \setminus v) \cup \{ z_e : e \in C \cap E(v) \}$. We make
  the following observations:
\begin{itemize}
\item $S$ is an $\pair{X}{v}$-separator in $G'$, as each $\pair{X}{v}$-path in $G$
  either goes through $C_Y \setminus v$, or goes through an edge of
  $C$ incident to $v$ (note also that $S$ is disjoint from $X\cup
  \{v\}$).
\item $u \in \textup{Reach}_{G'}(X,S)$: as $C$ is a minimal $\pair{X}{Y}$-cut, there has to be an $\pair{X}{u}$-path in $G$ disjoint from $C$, that is, fully contained in $U$, which means that the corresponding path in $G'$ avoids $S$.
\item $|S| \leq 6k$. By Lemma~\ref{lem:solspan}, the edges in $C$ span at most
  $6k$.  Every vertex of $C$ can appear in $C_Y$ or
  can be adjacent to $v$, but not both. Therefore, each vertex spanned
  by $C$ contributes at most one to $S$ and $|S|\leq 6k$ follows.
\end{itemize}
By the definition of important separators, there exists an 
important $\pair{X}{v}$-separator $S'$ in $G'$ such that $\textup{Reach}_{G'}(X,S)
\subseteq \textup{Reach}_{G'}(X,S')$ and $|S'| \leq |S|$. As $z_e$ is
adjacent to $u$ and $v$, as $u \in \textup{Reach}_{G'}(X,S')$ and as
$S'$ is an $\pair{X}{v}$-separator in $G'$, it follows that $z_e \in S'$. Now,
$S' \subseteq C_{6k}(X,\{v\})$ implies that $z_e \in C_{6k}(X,\{v\})$, and we
conclude that $e \in E_{\textup{rel}}(v)$. \qed
\end{proof}

We construct an instance $I'$ of \textsc{Constrained Rank-Cut} from the
instance $I$ of \textsc{Rank-Cut}, by the following random process. Let
$p = \frac{1}{6 k d}=2^{-O(k)}$. We color edges of $E_{\textup{rel}}\setminus M$
with color black with probability $p$, and with color red otherwise. Let $E_b$
denote the set containing the edges of $E_{\textup{rel}}$ colored black,
as well as the edges of $M$. Consider the subgraph $G_b$ of $G$ containing
only the edges in $E_b$ and let partition $P=(V_1,\ldots,V_{\ell'})$ represent the
way the connected components of this subgraph partition $V(G)$ (note that $P$
can have classes that contain only a single vertex). By definition, $G[V_i]$ is connected
for every $i$ and the two endpoints of each edge in $M$ is in the same $V_i$.
Therefore, the \textsc{Constrained Rank-Cut} instance $I' = (G,k,X,Y,M,P)$ is a
valid instance, as it satisfies both (i) and (ii).

\begin{lemma} \label{lem:randomized-reduction} The following two statements hold:
\begin{enumerate}
\item If $I$ is a no-instance of \textsc{Rank-Cut}, then $I'$ is a no-instance of \textsc{Constrained Rank-Cut}.
\item If $I$ is a yes-instance of \textsc{Rank-Cut}, then $I'$ is a yes-instance of \textsc{Constrained Rank-Cut} with probability $2^{-O(k^2)}$.
\end{enumerate}
\end{lemma}

\begin{proof} Clearly, if $I'$ has a solution $Z$, then $C_Z$ is a
  solution for instance $I$ of \textsc{Rank-Cut}. Conversely, suppose
  that $I$ has a minimal solution $C \subseteq E(G)$ with $r_{M}(C)
  \leq k$. Let $U_1,\ldots,U_{\ell'}$ denote the vertex sets of the
  connected components of $G[C \cup M]$ (note that this is not
  necessarily a partition of $V(G)$, as it is possible to have
  vertices that are not incident to any edge of $C\cup M$). Let $F$ be
  a spanning forest of $G[C \cup M]$ containing as many edges of $M$
  as possible. Let $B=F \setminus M$; as all these edges are in $C$,
  we have $B \subseteq E_{\textup{rel}}$ by Lemma
  \ref{lem:relevant-edges}, and since we have $r_M(C) \leq k$ it
  follows that $|B| \leq k$. Let $R= \cup_{i = 1}^{\ell'}
  E_{\textup{rel}}(U_i)$, where $E_{\textup{rel}}(U_i)$ denotes the
  set of edges in $E_{\textup{rel}}$ with exactly one endpoint in
  $U_i$. By Lemma~\ref{lem:solspan}, $C\cup M$ spans at most $6k$
  vertices, thus $\sum_{i=1}^{\ell'} |U_i|\le 6k$. As each vertex of
  $V(G)$ has at most $d$ incident edges in $E_{\textup{rel}}$, we have
  $|R| \leq d \sum_{i=1}^{\ell'} |U_i| \le 6 k d=2^{O(k)}$. Now, (i)
  with probability at least $p^k=2^{-O(k^2)}$, every edge of $B$ is
  colored black, (ii) with probability at least $(1-\frac{1}{6 k
    d})^{6 k d} \geq \frac{1}{4}$, every edge of $R$ is colored red
  (indeed, the function $x \mapsto (1-\frac{1}{x})^{x}$ is increasing
  on $[1,+\infty[$ and is thus $\geq \frac{1}{4}$ for $x \geq
  2$). These two events are independent, as they involve disjoint sets
  of edges. Suppose that both events happen. Then, $E_b$ contains all
  edges of $F$, but no edge of $R$. Consider the subgraph $G_b$ of $G$
  containing only the edges in $E_b$ and let partition
  $P=(V_1,\ldots,V_{\ell})$ represent the way the connected
  components of this subgraph partition $V(G)$.
 Then every $U_i$ is one
  class of this partition. Thus $C' = \cup_{i = 1}^{\ell'} E(G[U_i])$
  is a solution for instance $I'$ (as $C' \supseteq C \cup M$ and
  $r_M(C') = r_M(C) \leq k$). We conclude that $I'$ is a yes-instance
  with probability $2^{-O(k^2)}$.  \qed
\end{proof}

From Lemmas \ref{lem:algo-constrained-rank-cut} and \ref{lem:randomized-reduction}, we obtain:

\begin{theorem} \label{thm:rank-cut-randomized} \textsc{Rank-Cut} has a randomized \textup{FPT} algorithm with running time $2^{O(k^2)} n m$.
\end{theorem}

\begin{proof} Let $I = (G,k,X,Y,M)$ be an instance of \textsc{Rank-Cut}.
  We first remove all isolated vertices of $G$ in time $O(n+m)$,
  obtaining a graph $G$ for which each connected component has at least
  two vertices, ensuring that $n+m = O(m)$. We then compute the set
  $E_{\textup{rel}}$ in time $O(4^{6 k} k n m)$, and we construct the instance $I'$ of 
  \textsc{Constrained Rank-Cut} by random selection as described above. This instance
  $I'$ can be solved in time $O(k(n+ m))$ by Lemma \ref{lem:algo-constrained-rank-cut}.
  By Lemma~\ref{lem:randomized-reduction}, the probability of a correct
  answer is at least $p_\textup{correct} = 2^{-O(k^2)}$. Thus repeating this process
  $\lceil \frac{1}{p_\textup{correct}}\rceil = 2^{O(k^2)}$ times yields a randomized FPT
  algorithm for \textsc{Rank-Cut} running in time $2^{O(k^2)} n m$ and
  having success probability $(1-p_\textup{correct})^{\frac{1}{p_\textup{correct}}} \geq \frac{1}{4}$.  \qed
\end{proof}

\section{Derandomization}\label{sec:derandomization}

We now derandomize the proofs of Lemma \ref{lem:randomized-reduction}
and Theorem \ref{thm:rank-cut-randomized} using the standard technique
of splitters. Given integers $n,s,t$, an \emph{$(n,s,t)$-splitter} is
a family $\F$ of functions $f : [n] \rightarrow [t]$ such that for
every $S \subseteq [n]$ with $|S| = s$, there is a function of $\F$
that is injective of $S$. Naor et al.~\cite{NSS95} give a
deterministic construction of an $(n,s,s^2)$-splitter of size $O(s^6
\log s \log n)$. We can use this splitter construction to build a
family of colorings of $E_{\textup{rel}}$ to replace the randomized
selection of colors in Lemma~\ref{lem:randomized-reduction}. The proof
of Lemma~\ref{lem:randomized-reduction} analyzed the probability of
the event that a certain set $B$ of edges is colored black and at the
same time a certain set $R$ of edges is colored red.  By setting the
parameters of splitters appropriately, we can ensure that at least one coloring in
the family has this property. The $(n,s,s^2)$-splitter of Naor et
al.~\cite{NSS95} can be constructed in polynomial time, but
unfortunately the exact running time is not stated. Therefore, in the
following theorem, we do not specify the polynomial factors of the
running time.

\begin{theorem} \label{thm:rank-cut-deterministic} \textsc{Rank-Cut} has a deterministic \textup{FPT} algorithm with running time $2^{O(k^2)} n^{O(1)}$.
\end{theorem}
\begin{proof}
Consider an instance $I = (G,k,X,Y,M)$ of \textsc{Rank-Cut}. We first construct the set $E_{\textup{rel}}$ as in Section
\ref{sec:reduction-constrained-rank-cut}, and we identify
$E_{\textup{rel}}\setminus M$ with the set $[m']$ where $m' =
|E_{\textup{rel}}\setminus M|$. Let $s = k+4 k d = 2^{O(k)}$. Using the result of
\cite{NSS95}, we construct an $(m',s,s^2)$-splitter $\F$ of size
$O(s^6 \log s \log m)$. Instead of randomly coloring the elements of
$E_{\textup{rel}}\setminus M$, we go through the following
deterministic family of colorings: for every $f \in \F$ and every
subset $U \subseteq [s^2]$ of size at most $k$, we color $e\in
E_{\textup{rel}}\setminus M$ black if and only if $f(e)\in U$.  For
each such coloring, we perform the reduction to \textsc{Constrained
  Rank-Cut} as in Lemma~\ref{lem:randomized-reduction} and then solve
the instance using the algorithm of
Lemma~\ref{lem:algo-constrained-rank-cut}. We return ``yes'' if and
only if at least one of the resulting \textsc{Constrained Rank-Cut}
instances is a yes-instance.

It is clear that if one of the \textsc{Constrained Rank-Cut} instances
is a yes-instance, then $I$ is a yes-instance of
\textsc{Rank-Cut}. Conversely, suppose that $I$ is a yes-instance and
let $B$ and $R$ be the set of edges defined in the proof of
Lemma~\ref{lem:randomized-reduction}. As $|B|+|R|\le s$, there is a
function $f\in \F$ that is injective on $B\cup R$ and there is a set
$U\subseteq [s^2]$ of size at most $k$ such that $b\in B\cup R$
satisfies $b\in B$ if and only $f(b)\in U$.  For this choice of $f$
and $U$, the algorithm considers a coloring that colors $B$ black and
$R$ red. Therefore, the reduction creates a yes-instance of
\textsc{Constrained Rank-Cut}.\qed
\end{proof}

Theorems \ref{thm:rank-cut-randomized} and
\ref{thm:rank-cut-deterministic} respectively give randomized and
deterministic FPT algorithms for \textsc{Rank-Cut}. Combining them
with Lemmas \ref{lem:reduction-bcc} and \ref{lem:reduction-rank-cut},
we obtain (i) a $2^{O(k^2)} n m$ randomized algorithm for
\textsc{Bipartite Contraction}, (ii) a $2^{O(k^2)} n^{O(1)}$
deterministic algorithm \textsc{Bipartite Contraction}. This
establishes Theorem \ref{thm:algo-bc} stated in the introduction.

\section{Concluding remarks}\label{sec:conclusion}

We have obtained a randomized $2^{O(k^2)} n m$ algorithm for
\textsc{Bipartite Contraction}. Can the dependence on $k$ be improved?
It seems plausible that the problem admits a $2^{O(k)} n^{O(1)}$ FPT
algorithm, as such algorithms are known for \textsc{Edge
  Bipartization} \cite{GGHNW06} as well as for other edge contraction
problems \cite{HHLLP11}. We note that important separators are a
common feature of \cite{HHLP11} and of our algorithm, so they could be
the key to further improvements.

Regarding kernelization, Heggernes et al.~\cite{HHLP11} asked whether
\textsc{Bipartite Contraction} has a polynomial kernel. While this
question is still open, it is now known that \textsc{Odd Cycle
  Transversal} (and thus \textsc{Edge Bipartization}) have randomized polynomial
kernels \cite{KW12a}. As \textsc{Edge Bipartization} reduces to
\textsc{Bipartite Contraction}, this raises the question whether the
matroid-based techniques of \cite{KW12a,KW12b} can be applied to the
more general \textsc{Bipartite Contraction} as well. The notion of
rank in the \textsc{Rank-Cut} problem is the same as the notion of
rank in graphic matroids, hence it is possible that the rank
constraint can be incorporated into the arguments of
\cite{KW12a,KW12b} based on linear representation of matroids.

\bibliographystyle{abbrv}
\bibliography{bipartiteContraction}

\begin{thebibliography}{10}

\bibitem{AYZ95}
N.~Alon, R.~Yuster, and U.~Zwick.
\newblock Color-coding.
\newblock {\em J. ACM}, 42(4):844--856, 1995.

\bibitem{DBLP:journals/corr/abs-1211-5933}
Y.~Cao and D.~Marx.
\newblock {Interval Deletion is Fixed-Parameter Tractable}.
\newblock {\em CoRR}, abs/1211.5933, 2012.

\bibitem{CFLLV08}
J.~Chen, F.~V. Fomin, Y.~Liu, S.~Lu, and Y.~Villanger.
\newblock Improved algorithms for feedback vertex set problems.
\newblock {\em J. Comput. Syst. Sci.}, 74(7):1188--1198, 2008.

\bibitem{DBLP:journals/algorithmica/ChenLL09}
J.~Chen, Y.~Liu, and S.~Lu.
\newblock {An Improved Parameterized Algorithm for the Minimum Node Multiway
  Cut Problem}.
\newblock {\em Algorithmica}, 55(1):1--13, 2009.

\bibitem{CLLSR08}
J.~Chen, Y.~Liu, S.~Lu, B.~O'Sullivan, and I.~Razgon.
\newblock A fixed-parameter algorithm for the directed feedback vertex set
  problem.
\newblock {\em J. ACM}, 55(5), 2008.

\bibitem{dirmway-journal}
R.~H. Chitnis, M.~Hajiaghayi, and D.~Marx.
\newblock {Fixed-Parameter Tractability of Directed Multiway Cut Parameterized
  by the Size of the Cutset}.
\newblock To appear in {\em SIAM Journal on Computing.}
  http://arxiv.org/abs/1110.0259.

\bibitem{CPP12}
M.~Cygan, M.~Pilipczuk, and M.~Pilipczuk.
\newblock {On Group Feedback Vertex Set Parameterized by the Size of the
  Cutset}.
\newblock In {\em Proceedings WG 2012}, volume 7551 of {\em LNCS}, pages
  194--205, 2012.

\bibitem{DFLRS07}
F.~K. Dehne, M.~R. Fellows, M.~A. Langston, F.~A. Rosamond, and K.~Stevens.
\newblock {An $O(2^{O(k)} n^3)$ FPT algorithm for the undirected feedback
  vertex set problem}.
\newblock {\em Theor. Comput. Syst.}, 41(3):479--492, 2007.

\bibitem{GHP12}
P.~A. Golovach, P.~van't Hof, and D.~Paulusma.
\newblock {Obtaining Planarity by Contracting Few Edges}.
\newblock In {\em Proceedings MFCS 2012}, volume 7464 of {\em LNCS}, pages
  455--466, 2012.

\bibitem{G11}
S.~Guillemot.
\newblock {FPT algorithms for path-transversals and cycle-transversals
  problems}.
\newblock {\em Discrete Optimization}, 8(1):61--71, 2011.

\bibitem{GGHNW06}
J.~Guo, J.~Gramm, F.~H{\"u}ffner, R.~Niedermeier, and S.~Wernicke.
\newblock Compression-based fixed-parameter algorithms for feedback vertex set
  and edge bipartization.
\newblock {\em J. Comput. Syst. Sci.}, 72(8):1386--1396, 2006.

\bibitem{HHLLP11}
P.~Heggernes, P.~van't Hof, B.~L\'{e}v\^{e}que, D.~Lokshtanov, and C.~Paul.
\newblock {Contracting Graphs to Paths and Trees}.
\newblock In {\em Proceedings IPEC 2011}, volume 7112 of {\em LNCS}, pages
  55--66, 2011.

\bibitem{HHLP11}
P.~Heggernes, P.~van't Hof, D.~Lokshtanov, and C.~Paul.
\newblock {Obtaining a Bipartite Graph by Contracting Few Edges}.
\newblock In {\em Proceedings FSTTCS 2011}, volume~13 of {\em LIPIcs}, pages
  217--228, 2011.

\bibitem{DBLP:conf/focs/Kawarabayashi09}
K.~ichi Kawarabayashi.
\newblock {Planarity Allowing Few Error Vertices in Linear Time}.
\newblock In {\em FOCS}, pages 639--648, 2009.

\bibitem{DBLP:journals/corr/IwataOY13}
Y.~Iwata, K.~Oka, and Y.~Yoshida.
\newblock {Linear-Time FPT Algorithms via Network Flow}.
\newblock {\em CoRR}, abs/1307.4927, 2013.

\bibitem{MR2000c:68068}
H.~Kaplan, R.~Shamir, and R.~E. Tarjan.
\newblock {Tractability of parameterized completion problems on chordal,
  strongly chordal, and proper interval graphs}.
\newblock {\em SIAM J. Comput.}, 28(5):1906--1922, 1999.

\bibitem{KR10}
K.~Kawarabayashi and B.~A. Reed.
\newblock {An (almost) Linear Time Algorithm for Odd Cycle Transversal}.
\newblock In {\em Proceedings SODA 2010}, pages 365--378, 2010.

\bibitem{KW12a}
S.~Kratsch and M.~Wahlstr{\"o}m.
\newblock {Compression via Matroids: a Randomized Polynomial Kernel for Odd
  Cycle Transversal}.
\newblock In {\em Proceedings SODA 2012}, pages 94--103, 2012.

\bibitem{KW12b}
S.~Kratsch and M.~Wahlstr{\"o}m.
\newblock {Representative sets and irrelevant vertices: New tools for
  kernelization}.
\newblock In {\em Proceedings FOCS 2012}, 2012.
\newblock To appear.

\bibitem{DBLP:journals/corr/abs-1203-0833}
D.~Lokshtanov, N.~S. Narayanaswamy, V.~Raman, M.~S. Ramanujan, and S.~Saurabh.
\newblock {Faster Parameterized Algorithms using Linear Programming}.
\newblock {\em CoRR}, abs/1203.0833, 2012.

\bibitem{LSS09}
D.~Lokshtanov, S.~Saurabh, and S.~Sikdar.
\newblock {Simpler Parameterized Algorithm for OCT}.
\newblock In {\em Proceedings IWOCA 2009}, volume 5874 of {\em LNCS}, pages
  380--384, 2009.

\bibitem{M06}
D.~Marx.
\newblock Parameterized graph separation problems.
\newblock {\em Theoretical Computer Science}, 351(3):394--406, 2006.

\bibitem{M10}
D.~Marx.
\newblock Chordal deletion is fixed-parameter tractable.
\newblock {\em Algorithmica}, 57(4):747--768, 2010.

\bibitem{MSR12}
D.~Marx, B.~O'Sullivan, and I.~Razgon.
\newblock {Finding small separators in linear time via treewidth reduction}.
\newblock {\em ACM Transactions on Algorithms}, 2012.
\newblock To appear.

\bibitem{MS12}
D.~Marx and I.~Schlotter.
\newblock Obtaining a planar graph by vertex deletion.
\newblock {\em Algorithmica}, 62(3--4):807--822, 2012.

\bibitem{NSS95}
M.~Naor, L.~J. Schulman, and A.~Srinivasan.
\newblock Splitters and near-optimal derandomization.
\newblock In {\em Proceedings FOCS 1995}, pages 182--191, 1995.

\bibitem{DBLP:journals/corr/abs-1304-7505}
M.~S. Ramanujan and S.~Saurabh.
\newblock {Linear Time Parameterized Algorithms via Skew-Symmetric Multicuts}.
\newblock {\em CoRR}, abs/1304.7505, 2013.

\bibitem{RSV04}
B.~A. Reed, K.~Smith, and A.~Vetta.
\newblock Finding odd cycle transversals.
\newblock {\em Oper. Res. Lett.}, 32(4):299--301, 2004.

\bibitem{DBLP:conf/wg/BevernKMN10}
R.~van Bevern, C.~Komusiewicz, H.~Moser, and R.~Niedermeier.
\newblock {Measuring Indifference: Unit Interval Vertex Deletion}.
\newblock In {\em WG}, pages 232--243, 2010.

\bibitem{DBLP:journals/algorithmica/HofV13}
P.~van~'t Hof and Y.~Villanger.
\newblock {Proper Interval Vertex Deletion}.
\newblock {\em Algorithmica}, 65(4):845--867, 2013.

\bibitem{DBLP:journals/siamcomp/VillangerHPT09}
Y.~Villanger, P.~Heggernes, C.~Paul, and J.~A. Telle.
\newblock {Interval Completion Is Fixed Parameter Tractable}.
\newblock {\em SIAM J. Comput.}, 38(5):2007--2020, 2009.

\end{thebibliography}

\end{document}